\documentclass{jgcc}

\pdfoutput=1
\usepackage{lastpage}
\jgccdoi{12}{2}{2}{6649}

\jgccheading{}{\pageref{LastPage}}{}{}%
{Jul.~18,~2020}{Feb.~12,~2021}{}

\usepackage[T1]{fontenc}
\usepackage[utf8]{inputenc}
\usepackage{lmodern}
\usepackage{tikz}
\usetikzlibrary{shapes.geometric, arrows}
\tikzstyle{startstop} = [rectangle, rounded corners, minimum width=3cm, minimum height=1cm,text centered, draw=black]
\tikzstyle{io} = [trapezium, trapezium left angle=70, trapezium right angle=110, minimum width=3cm, minimum height=1cm, text centered, draw=black]
\tikzstyle{process} = [rectangle, minimum width=3cm, minimum height=1cm, text centered, draw=black]
\tikzstyle{SubProcess} = [rectangle, minimum width=1cm, minimum height=1cm, text centered, draw=black]
\tikzstyle{decision} = [diamond, minimum width=3cm, minimum height=1cm, text centered, draw=black]
\tikzstyle{arrow} = [thick,->,>=stealth]
\usetikzlibrary{decorations.pathreplacing,matrix}
\usepackage{float}
\usepackage{graphicx}
\usepackage[ruled,vlined]{algorithm2e}
\usepackage{listings}
\usepackage{cite}
\newtheorem{theorem}{Theorem}
\newtheorem{lemma}{Lemma}

\newtheorem*{conjecture}{Conjecture A}

\newtheorem{definition}{Definition}[section]
\makeatletter
\renewcommand*\env@matrix[1][*\c@MaxMatrixCols c]{%
  \hskip -\arraycolsep
  \let\@ifnextchar\new@ifnextchar
  \array{#1}}
\makeatother
\title[Elliptic curve discrete logarithm problem]{A new method for solving the elliptic curve discrete logarithm problem}
\author{Ansari Abdullah}
\address{Savitribai Phule Pune University, Pune, India}
\email{abdullah0096@gmail.com}
\author{Ayan Mahalanobis}
\address{IISER Pune, Pune, India}
\email{ayan.mahalanobis@gmail.com}
\thanks{Second author was supported by a SERB MATRICS grant and a NBHM research grant.}
\author{Vivek M.~Mallick}
\address{IISER Pune, Pune, India}
\email{vmallick@iiserpune.ac.in}	
\date{}
\begin{document}	
\begin{abstract}
The elliptic curve discrete logarithm problem is considered a secure
cryptographic primitive. The purpose of this paper is to propose a paradigm
shift in attacking the elliptic curve discrete logarithm problem. In this
paper, we will argue that initial minors are a viable way to solve this
problem. This paper will present necessary algorithms for this attack. We
have written a code to verify the conjecture of initial minors using Schur
complements. We were able to solve the problem for groups of order up to
$2^{50}$. 
\end{abstract}
\maketitle
\section{Introduction}
The discrete logarithm problem on rational points of an elliptic curve
(ECDLP) has been of interest to cryptographers for the last four decades. It is still considered a secure cryptographic primitive and is used in the latest \emph{transport layer security protocol} (TLS 1.3).
In this paper we present a new way to attack ECDLP. We present this new way
by formulating a conjecture. This paper is an extension of our earlier work~\cite{one}. We will keep using the same notations to maintain continuity. Let $\mathcal{E}$ be a non-singular elliptic curve over the finite field $\mathbb{F}_q$ in which the ECDLP will be solved.

Fix a positive integer $n^\prime$ and let $k=3n^\prime$. The \emph{Las Vegas
algorithm} that we developed in\cite{one}, tries to find a planar projective
curve $\mathcal{C}$ of degree $n^\prime$ passing through $P_i$,  $1 \leq i
\leq k$ where each point $P_i$ is on the elliptic curve $\mathcal{E}$. If
such a curve $\mathcal{C}$ exists, ECDLP is solved. This algorithm belongs
to the class of non-generic algorithms because it depends on certain
properties of the elliptic curve and cannot be used to solve the discrete
logarithm problem in other groups. Another example of a non-generic
algorithm is the index calculus algorithm to solve the discrete logarithm
problem in finite fields. Our algorithm has no restriction on the finite
field on which the elliptic curve is defined.

The Las Vegas algorithm (Algorithm \ref{alg:lasvegas}, and also see our
earlier work \cite{one}) that we developed reduces ECDLP to a linear algebra
problem. In our earlier work, we referred to it as Problem L and will
continue to do so in this work. We write this paper to \textbf{propose a
conjecture}, we will call it \emph{initial minors}.

We first show that one way to solve Problem L (and thus ECDLP) is to find a
zero-minor in a non-singular matrix. Now the question arises, how to find a
zero-minor? We argue that there is a set of minors, such that, if every
minor in that set is non-zero then all minors of the matrix are non-zero.
Then the complexity of the algorithm to solve ECDLP reduces to the
cardinality of the set of initial minors. Conjecturally, the cardinality of
the set of initial minors should be sub-exponential with respect to the size
of the matrix, whereas the cardinality of the set of all minors is
exponential.  The idea of initial minors is not new, it has been well
studied for total positive matrices~\cite{ando} and has given rise to
important mathematics~\cite{icm}.

\begin{definition}
  A set of minors for a non-singular matrix $\mathcal{A}$ is called a
  \emph{set of initial minors} if the following property holds: every minor
  in the set of initial minors is non-zero if and only if all minors in
  $\mathcal{A}$ are non-zero.
\end{definition}

\begin{conjecture}[Initial Minors]
  For a matrix $\mathcal{A}$ coming from the above approach to solve ECDLP,
  there exists a set of initial minors, whose cardinality is bounded by a
  sub-exponential function of the size of the matrix.
\end{conjecture}

A stronger conjecture can be obtained by replacing sub-exponential by
polynomial above. But since the conjecture is based on experimental data,
let us state it in the modest form above.

\noindent\textbf{Acknowledgement:} We would like to thank the anonymous
referee for the valuable comments and suggestions, leading to a
considerably improved exposition of the material.

\section{The Las Vegas algorithm}
 With a slight abuse of notation we will denote by $\mathcal{E}$ the group
 of rational points of $\mathcal{E}$ over $\mathbb{F}_q$ as well.
 Furthermore, we assume that the group $\mathcal{E}$ is of prime order $p$.
 Let $P,Q\in\mathcal{E}$ be non-zero elements, \emph{the elliptic curve
 discrete logarithm problem} is to find an integer $m$, where $1\leq m<p$,
 such that $Q=mP$. Note, since $\mathcal{E}$ is a group of prime-order, such
 integers $m$ exist.

The central idea behind the attack is presented below as Theorem
\ref{the:mainthem}. The method is to construct a matrix $\mathcal{M}$ and
compute its left-kernel $\mathcal{K}$. We construct $\mathcal{M}$ with $k$
rows, where $k=3n^\prime$, one row at a time. Consider a planar projective
curve $\mathcal{C}=\sum_{u+v+w=n^\prime}a_{u,v,w}x^uy^vz^w $ where
$a_{u,v,w}\in\mathbb{F}_q$. Clearly $x^{u}y^{v}z^{w}$ ranges over all
possible monomials of degree $n^\prime $ which we assume to be ordered. For
$P_0\in\mathcal{E} $, say  $P_0=(x_0: y_0: z_0)$, we construct a row in
$\mathcal{M} $ where $(x_0:y_0:z_0)$ is substituted for $x,y$ and $z$ in
$\mathcal{C}$. The ordering in the row is the same as the ordering of
monomials in $\mathcal{C}$. We construct $\mathcal{M}$ by adding rows to it
corresponding to each $P_i$, as we did in our example for $P_0$. These $P_i$
are random distinct points on $\mathcal{E}$ constructed from random distinct
integers $n_i;\,0<n_i<p$ by computing $n_iP$. Then we compute $\mathcal{K}$
as the left-kernel of ${\mathcal{M}}$. Now if we look at the (right) kernel
$\mathcal{K}^\prime$ of $\mathcal{M}$. The elements of $\mathcal{K}^\prime$
actually produce the curve $\mathcal{C}$ that we mentioned earlier. It
clearly passes through all the points $P_i$ from which $\mathcal{M}$ was
constructed. 

A fundamental idea behind this attack rests on the fact that the left-kernel
$\mathcal{K}$ being nonzero is equivalent to the existence of a planar
projective curve $C$ in the right kernel of $\mathcal{M}$ which does not
contain the elliptic curve $\mathcal{E}$. The proof is in our earlier
work~\cite[Theorem 2]{one}.

A question can be raised: 
Assume that we constructed $\mathcal{M}$ as above and then $\mathcal{K}\neq
0$. This says that there is a curve $C$ of degree $n^\prime$. Can this curve
$C$ pass through a point $P$ on the elliptic curve from which $\mathcal{M}$
was constructed where the  intersection multiplicity at $P$ on $C$ and
$\mathcal{E}$ is more than one? If that happens then this curve $C$ passes
through more than $3n^\prime$ point on $\mathcal{E}$ counting intersection
multiplicities, because it is in the right kernel of $\mathcal{M}$ with
$3n^\prime$ points. Thus it intersects the elliptic curve in more than
$3n^\prime$ points counting intersection multiplicities and is of degree
$n^\prime$. By B\'{e}zout's theorem it contains the elliptic curve
$\mathcal{E}$ which contradicts $\mathcal{K}\neq 0$. This proves the
following theorem:
\begin{theorem}
  \label{the:intmult1}
If $\mathcal{K}\neq 0$, every non-zero vector in $\mathcal{K}$ corresponds
to a planar projective curve $C$ of degree $n^\prime$ over $\mathbb{F}_q$.
This curve does not contain the elliptic curve $\mathcal{E}$ and intersects
the elliptic curve in $3n^\prime$ distinct points with intersection
multiplicity one. 
\end{theorem} 

\begin{proof}
  Since $\mathcal{K} \neq 0$, the fact from \cite[Theorem 2]{one} stated
  above, ensures a planar projective curve $C$ of degree $n'$.

  To complete the proof, we need to prove that the curve $C$ passes through
  each chosen point on the elliptic curve with multiplicity exactly once. If
  this does not hold, the curve passes through one of the chosen points with
  multiplicity at least 2.  However, being an element of the right kernel,
  it does pass through all the $3n'$ chosen points. This means that $C$ and
  $\mathcal{E}$ intersect at more than $3n'$ points contradicting
  B{\'e}zout's theorem.
\end{proof}

The next theorem is vital for the algorithm we develop. Now there are two positive integers $k$ and $l$. We think of $l$ as extra points and $k=3n^\prime$, where $n^\prime$ is a positive integer. The choice of $3n^\prime$ comes from B\'{e}zout's theorem, that a degree $n^\prime$ planar projective curve intersects an elliptic curve at most at $3n^\prime$ points.

The inclusion of extra points lets us test many different groups of points
of size $k$ simultaneously. In particular, we can test for ${k+l\choose l}$
possible curves $\mathcal{C}$ simultaneously. This is a basic advantage of
our algorithm. We have shown that (see \cite[Section 3]{one}) $k=l$ gives us
the optimal probability of success for our Las Vegas algorithm. The size of
the elliptic curve group $\mathcal{E}$ is taken large enough compared to $k$
and $l$ such that our theorems make sense. For practical purposes, the size
of $\mathcal{E}$ is a prime $p$ and $k,l$ is about the size of $\log_2{p}$. 
\subsection{The main theorem}
\begin{theorem}[Main Theorem]
  \label{the:mainthem}
Let $k=3n^\prime$ for some positive integer $n^\prime$. Let $l$ be another positive integer. Choose positive integers $s$ and $t$, such that, $s\neq t$ but $s+t=k+l$. Then construct the matrix $\mathcal{M}$ as described before with rows corresponding to $n_iP$ for $i=1,2,\ldots,s$ and $-n_jQ$ for $j=1,2,\ldots,t$.

Let $\mathcal{K}$ be the left-kernel of $\mathcal{M}$. The following is true:
\begin{description}
\item[a] The left-kernel $\mathcal{K}$ is of dimension $l$.
\item[b] If there is a vector $v$ in $\mathcal{K}$ with at least $l$ zeros then there is a curve $\mathcal{C}$ passing through $3n^\prime$ points corresponding to the non-zero points of $v$.
\end{description}
\end{theorem}
\begin{proof}
  A detailed proof of this theorem is in our earlier work~\cite[Corollary 1]{one}. For the convenience of the reader we will sketch a rough argument. The basic reason behind the argument is that $\sum_{i=1}^kP_i=\mathcal{O}$ for $k$ points on the elliptic curve if and only if there is a planar projective curve $C$ of degree $n^\prime$ passing through these points. For a proof look at~\cite[Theorem 1]{one}. Then the proof of (a) above follows from simple counting argument~\cite[Theorem 3]{one}.

For a proof of (b) notice that if there is a vector in the left-kernel with at least $l$ zeros then there is a curve passing through the points on the curve that correspond to the non-zero elements of the vector. There are at most $3n^\prime$ points.
\end{proof}

Using the above theorem, the algorithm is broken down into two parts. One,
is a randomized algorithm that creates the left-kernel. The second part is
to determine if the left-kernel contains a vector with $l$ zeros. We called
this second part Problem L in \cite{one}.  If there is a vector with
$l$ zeros, we have to find that vector because the position of these zeros
matter. In this paper, we convert this problem into finding a zero-minor.
Note that, the probability is optimum when $k=l$ and
$k=O(\log{p})$~\cite[Theorem 5]{one}. So we will assume that $k=l$ for the
rest of the paper.  Note that now $\mathcal{K}$ is a $k$-dimensional
subspace of the $2k$-dimensional vector space over $\mathbb{F}_q$. The rows
of which are vectors of length $2k$ over $\mathbb{F}_q$. We define Problem L
as follows:
\begin{definition}[Problem L]
Determine if the subspace generated by the rows of $\mathcal{K}$ contains a vector with $k$ zeros. If it does, find the location of these zeros. 
\end{definition}
\subsection{Using Theorem \ref{the:mainthem}}
To use Theorem \ref{the:mainthem}, first we select an integer $n^\prime=O(\log{p})$ and compute $k = 3n^\prime$. Then choose $2k$ random \emph{distinct positive integers} of size less than $p$ and form a sequence $S$ of size $s+t$.
Using elements from $S$ compute points $P_i=n_iP$ for $i=1,2,\ldots,s$ and $Q_j=-n_{s+j}Q$ for $j=1,2,\ldots,t$.

Now we check to see if there exist a planar projective curve $\mathcal{C}$ of degree $n^\prime$ passing through any subset of size $3n^\prime$ of $s+t$ points. If such a curve exists, the discrete logarithm problem is solved. Else a new set $S$ is generated and the process is repeated.

In order to prove the existence of a homogeneous curve of degree $n^\prime$ we construct a matrix $\mathcal{M}$, corresponding to $P_i$ and $Q_j$ as rows.

Each row in $\mathcal{M}$ represents a $P_i$ or a $Q_j$. In other words, each row represents an integer from $S$. We describe the process as an algorithm.

\begin{algorithm}[H]
  \label{alg:lasvegas}
\SetAlgoLined
\:
\KwData{Two points $P,Q\in\mathcal{E}$ such that $mP = Q$. Order of $\mathcal{E}$ is a prime $p$. }
\KwResult{$m$}
\vspace{1ex}
  \hspace{1ex}\textbf{Step 1 :} $n^\prime \gets O\left(\log_{2}(p)\right)$\\
  \hspace{1ex}\textbf{Step 2 :} $k \gets 3n^\prime$\\
  \hspace{1ex}\textbf{Step 3 :} Generate $2k$ random positive integers less than $p$\\
  \hspace{1ex}\textbf{Step 4 :} Use random numbers and $\mathcal{C}$ along with $P$ and $Q$ to generate the matrix $\mathcal{M}$ using points from above\\
  \hspace{1ex}\textbf{Step 5 :} $\mathcal{K} \gets \text{left-kernel} (\mathcal{M}$)\\
  \hspace{1ex}\textbf{Step 6 :} If $\mathcal{K}$ has a vector with $k$ zeros stop else go to Step 3\\
\caption{A Las Vegas algorithm to reduce ECDLP to Problem L}
\end{algorithm} 
Note that we generate $2k$ points above because we take $l=k$. We will soon move from $\mathcal{K}$ to $\mathcal{A}$ which is a $k\times k$ matrix consisting of the non-identity part of $\mathcal{K}$.
\subsection{Complexity of the Las Vegas algorithm}
The Las Vegas algorithm can be divided into two parts.
The first part reduces the ECDLP to a problem in linear algebra. The second part is to solve the Problem L. Reducing ECDLP to Problem L is a Las Vegas algorithm with success probability 0.6 and it has polynomial 
space and time complexities~\cite[Theorem 5]{one}. 
\section{Minors to solve ECDLP}
Let $A$ be a $k\times k$ non-singular matrix. Let $\alpha,\beta$ be non-empty subsets of $\{1,2,\ldots,k\}$ of same size. Then $A[\alpha|\beta]$ is a square sub-matrix of the size of $\alpha$ and $\beta$
which contains elements that are in the intersection of rows and columns of $A$ indexed by $\alpha$ and $\beta$ respectively. The determinant $\det\left(A[\alpha|\beta]\right)$ is the minor $A[\alpha|\beta]$.
A principal minor is a minor when $\alpha=\beta$ and is denoted by $A[\alpha]$. 
In this paper, we are interested in finding $\alpha$ and $\beta$ such that $\det\left(A[\alpha|\beta]\right)=0$.
\subsection{Using minors to solve problem $L$}
Let $\mathcal{K}$ be a $k\times 2k$ matrix. Now the rows of $\mathcal{K}$ is a basis of the left-kernel of $\mathcal{M}$. One can do a row-reduction of $\mathcal{K}$ to reduce one part of it to the identity. We will assume that $\mathcal{K}=\mathcal{A}|I$. Where $\mathcal{A}$ is a $k\times k$ matrix and $I$ is the identity matrix of size $k$. Note that, in NTL~\cite{Shoup1996NtlAL}, $\mathcal{K}$ is produced as follows:
\begin{equation*}
\mathcal{K} =   \begin{pmatrix}[ccccc|ccccc]
    a_{11} & a_{12} & a_{13} & \dots  & a_{1k} & 0 &0&\dots&0&1 \\
    a_{21} & a_{22} & a_{23} & \dots  & a_{2k} &0 &0&\dots&1&0 \\
    \vdots & \vdots & \vdots & \ddots & \vdots &\vdots & \vdots & \ddots & \vdots & \vdots \\
    a_{k1} & a_{k2} & a_{k3} & \dots  & a_{kk} &1 &0&\dots&0&0
\end{pmatrix}
\end{equation*}
Now each row of $\mathcal{K}$ contains $k-1$ zeros. Our problem is to find a vector which is a linear combination of the rows of $\mathcal{K}$ which has $k$ zeros. We would be more interested in the
non-identity
part of $\mathcal{K}$ which we denote by $\mathcal{A}$.
\begin{equation*}
\mathcal{A} =   \begin{pmatrix}[ccccc]
    a_{11} & a_{12} & a_{13} & \dots  & a_{1k}  \\
    a_{21} & a_{22} & a_{23} & \dots  & a_{2k}  \\
    \vdots & \vdots & \vdots & \ddots & \vdots  \\
    a_{k1} & a_{k2} & a_{k3} & \dots  & a_{kk}  
\end{pmatrix}
\end{equation*}
Note that, if $a_{ij}=0$ for some $i,j$ then ECDLP is solved. We now prove a lemma leading to a theorem.
\begin{lemma}
  \label{lem:atlstvzr}
  Recall the definition of $\mathcal{K}$. If there is a vector $v$ in
  $\mathcal{K}$ with at least $k$ zeros, then $v$ has exactly $k$ zeros.
\end{lemma}
\begin{proof}
Assume for a moment that $v$ has more than $k$ zeros.
Then concentrate on the non-zero entries of $v$. Assume that there are $k^\prime$ of them which is less than $k$. Now we create a new matrix $\mathcal{M}^\prime$.
First $k^\prime$ rows of $\mathcal{M}^\prime$ are the rows in $\mathcal{M}$
with indices same as indices of non-zero entries in $v$. Now we can take any
arbitrary $k-k^\prime$ rows of $\mathcal{M}$ where $v$ is zero and put them
in $\mathcal{M}^\prime$. Then $\mathcal{M}^\prime$ has $k(=3n^\prime)$ rows
and a non-zero left-kernel. Then sum of this points of the elliptic curve is
$\mathcal{O}$, the point at infinity. However there are $k-k^\prime$
arbitrary rows in $\mathcal{M}^\prime$. We can replace any one of them with
an arbitrary element of the elliptic curve and get the sum of those points
as $\mathcal{O}$. This is a contradiction. Thus $v$ has exactly $k$ zeros.  
\end{proof}
\begin{theorem}
\label{pythagorean}
If $\det\left(\mathcal{A}[\alpha|\beta]\right) = 0$ for some non-empty subsets $\alpha,\beta\subseteq\{1,2,\ldots,k\}$, there exists a vector with $k$ zeros in the linear span of the rows of the left-kernel $\mathcal{K}$. Furthermore, the position of zeros are positions $\beta$ and $\{k+i:\,i\notin\alpha\}$.

Conversely, assume that there is a vector in the linear span of the rows of $\mathcal{K}$ with $k$ zero components, then there is a zero-minor in $\mathcal{A}$.
\end{theorem}
\begin{proof}
  Let us assume that there is $\alpha,\beta$ such that
  $\det\left(\mathcal{A}[\alpha|\beta]\right) = 0$. There is a sequence of
  row-operations on the sub-matrix $\mathcal{A}[\alpha|\beta]$ such that one
  row becomes zero. Now apply the same row-operations on $\mathcal{K}$ and
  obtain the vector $w$. We have zeros in components of $w$ in positions
  corresponding to elements of $\beta$. The non-zero components in positions
  greater than $k$ of $w$ depends on the rows used in the row-operation. The
  number of rows is at most the size of $\alpha$. But from Lemma
  \ref{lem:atlstvzr}, it follows that the number of rows must be same as
  $\alpha$. Moreover, the row operations cannot inadvertently introduce
  zeros in some other components of $w$ from positions $1$ to $k$. Taken all
  these together $w$ has $k$ zeros whose positions follows from the above
  discussion.

  Conversely, assume that there is a vector $w$ with $k$ zeros in the span
  of the rows of $\mathcal{K}$. Then there is a linear combination of rows
  of $\mathcal{A}$ that makes some components of $w$ in $\{1,2,\ldots,k\}$
  zero. Now consider those rows and indices of those rows constitute the set
  $\alpha$ and the indices in $\{1,2,\ldots,k\}$ of $w$ where the components
  are zero constitutes $\beta$. The size of $\alpha$ equals the size of
  $\beta$  from the above lemma and we have the zero-minor
  $\det\left(\mathcal{A}[\alpha|\beta]\right)$.

\end{proof}
The following example illustrates Theorem~\ref{pythagorean}.
Consider a $5 \times 10$ matrix $\mathrm{A}$ over $\mathbb{F}_{73}$ where each row has four zeros. We write

\[ \mathrm{A} = \left[ \:
\begin{array}{*{12}{c}}
\cline{1-2}
\multicolumn{1}{|c}{70} & \multicolumn{1}{c|}{18} & 1 & 17 & 10 & 0 & 0 & 0 & 0 & 1\\
\multicolumn{1}{|c}{10} & \multicolumn{1}{c|}{13} & 54 & 43 & 48 & 0 & 0 & 0 & 1 & 0\\
\cline{1-2}
23 & 43 & 8 & 24 & 57 & 0 & 0 & 1 & 0 & 0\\
29 & 29 & 56 & 61 & 48 & 0 & 1 & 0 & 0 & 0\\
49 & 38 & 21 & 46 & 27 & 1 & 0 & 0 & 0 & 0\\

\end{array}
\right]. \]
Note that $\mathrm{A}$ has a zero $2$-minor, for example, the minor

\begin{equation*}
{\mathrm{A}[\alpha|\beta] =  \det \begin{vmatrix}
70 & 18 \\
10 & 13 
\end{vmatrix}  = 0}
\end{equation*}
where
\begin{equation*}
\alpha = \{1, 2\}\;\text{and}\;\beta = \{1, 2\}.
\end{equation*}

The first row of $\mathrm{A}$ can be reduced by the row-operation $R_1 - (7 \times R_2)$ to yield

\[ \mathrm{A}_1 = \left[ \:
\begin{array}{*{12}{c}}
0 & 0 & 61 & 8 & 39 & 0 & 0 & 0 & 66 & 1\\
10 & 13  & 54 & 43 & 48 & 0 & 0 & 0 & 1 & 0\\
23 & 43 & 8 & 24 & 57 & 0 & 0 & 1 & 0 & 0\\
29 & 29 & 56 & 61 & 48 & 0 & 1 & 0 & 0 & 0\\
49 & 38 & 21 & 46 & 27 & 1 & 0 & 0 & 0 & 0\\
\end{array}
\right] \] 
The first row now has five zeros which solves ECDLP.
Thus if there is a zero-minor in $\mathrm{A}$, ECDLP is solved.
\section{Schur complement}
The Schur complement has a long and distinguished history in the theory of
matrices. Other than giving an algorithmic description of a variant of the
Schur complement (see equation \eqref{eqn:schurcpt} and the sentence
following that), we will not go into further details in this paper but will
use the work of Ando~\cite{ando} and Brualdi and Schneider~\cite{bru} as the
standard reference for this paper. Before describing the process of
generating the Schur complement, we would like to alert the reader
that ours is not the usual Schur complement~\cite[\S 1]{bru}. In the
traditional setting of Schur complements, matrices $E$ and $F$ remain fixed.
This is not so in our case.

Let $A=(a_{ij})$ be a $k\times k$ matrix. We can write $A$ as
\[\begin{pmatrix}
E & F\\
G & H
\end{pmatrix}\]
as block matrix where $E$ is the leading principal sub-matrix of $A$ of size
$k^\prime\times k^\prime$. In this section we assume that $E$ is
non-singular. The matrix $G$ is of size $(k-k^\prime)\times k^\prime$, $F$
is of size $k^\prime\times(k-k^\prime)$ and $H$ is of size
$(k-k^\prime)\times(k-k^\prime)$.

We look at the Schur complement in an algorithmic way. We talk about
row-operations in $E$, but they are row operations in the whole of $A$. We
do this for sake of exposition.  There are two different kind of
row-operations that we talk about one is in $E$ and another is by $E$.

We consider the first column of $A$. If $a_{11}=0$ we do a row permutation
in $E$ to get a non-zero $a_{11}$. We can do that because $E$ is
non-singular. Now using this $a_{11}$ and row-operations by $E$, we can make
all $a_{x1}=0$ for $x=2,3,\ldots,k$. Then we look at $a_{22}$. If it is zero
then we do a row-permutation in $E$ to make it non-zero. Then use $a_{22}$
and row operations by $E$ to make $a_{x2}=0$ for all $x=3,4,\ldots,k$. We
continue with $k^\prime$ columns of $E$ till we get a upper-triangular
matrix $E^\prime$ in place of $E$. In practice, if we do not get a upper
triangular matrix, i.e., one of the diagonal element is zero, we have solved
Problem L. Now the resulting matrix is of the form
\begin{equation} \label{eqn:schurcpt}
  A^\prime=
  \begin{pmatrix}
    E^\prime & F^\prime\\
    0 & H^\prime
  \end{pmatrix}.
\end{equation}
The matrix $H^\prime$ is called the \emph{Schur complement} and is denoted
by $A/E$. Furthermore, if $\det(E)\neq 0$ then $\det(E^\prime)\neq 0$.
Moreover $\det(A)=\det(A^\prime)$ and
$\det(A^\prime)=\det(E^\prime)\det(H^\prime)$.

Two properties of the Schur complement follow from this discussion. For sake of readability, we use
$I=\{k^\prime+1\leq i_1<i_2,\cdots,<i_r\leq k\}$ and $J=\{k^\prime+1\leq j_1<j_2,\cdots,<j_r\leq k\}$ for some positive integer $r$.
The first property is
\begin{equation}\label{eqn1}
\det\left(H^\prime[I|J]\right) = 0 \;\text{implies that}\;
\det\left(A[1,2,\ldots,k^\prime,I|1,2,\ldots,k^\prime,J]\right) = 0.
\end{equation}
We will use this to search for zero-minors in our algorithm.
The second property is that Schur complement of a Schur complement is a
Schur complement. This follows from the algorithmic way we defined Schur
complements, in particular, when we reduce $E$ to an upper-triangular matrix
systematically by reducing one column after another.

\subsection{An algorithm using the Schur complement}
The output of our randomized algorithm is a $k\times 2k$ matrix. The rows of this matrix is a basis for the left-kernel $\mathcal{K}$. So one can do a row operation on one half of this and make that half an identity matrix. We denote the other half, the  non-identity half by $\mathcal{A}$. For sake of definiteness we agree that this non-identity half is the first half of the left-kernel $\mathcal{K}$.

The question now is the following: How do we determine if there is a zero-minor in $\mathcal{A}$? There are too many minors. The number of principal minors is $2^k-1$.
It is not possible to check all of them. The question that we raise in this paper is: do we need to check all minors to determine if there is one minor of $\mathcal{A}$ that is zero? It is clear that the number of minors we check to determine if $\mathcal{A}$ has a zero-minor determines the complexity of our algorithm.

In the rest of this section, we take $k=3\log(p)$.

The first step of the algorithm is to compute $\mathcal{A}$ as defined
earlier and then check all possible 2-minors. If we find a zero 2-minor we
stop. 

If we find none, our first attempt was to compute a new $\mathcal{K}$ and
then $\mathcal{A}$. Repeat until a zero 2-minor is found. We count the
number of tries we had to go through until we find the first zero 2-minor,
see Table~\ref{tab:2-minor}. The average is over forty tries to solve ECDLP
over the same curve and the same group.
\begin{table}[!h]
\centering
\caption{ All 2-minors : Average number of iteration to solve ECDLP.} \label{tab:2-minor}
\vspace{2ex}
\begin{tabular}{l|c|c|c}
\hline
$\mathbb{F}_q$ & Average Iterations & $log_2{p} $ & No. of rows of $\mathcal{A}$  \\
\hline
$2^{40}$ & $113.4$ & 39 & 360  \\
$2^{43}$ & $716.3$ & 42 & 387  \\
$2^{46}$ & $4406.26$ & 45 & 414  \\
\hline
\end{tabular}
\end{table}

Hoping to get better results, we modified the algorithm by introducing Schur
complements. We did a column-reduce of one column after another, this
constructed one Schur complement after another. Recall that
$A=\begin{pmatrix}
E & F\\
G & H
\end{pmatrix}$, where $E$ is a principal diagonal matrix of size $k^\prime$. We start with $k^\prime=1$ and go until $k^\prime=\lfloor{k/2}\rfloor$. We reduce $E$ to a upper-triangular matrix by successively
reducing $G=0$ using row-operations. This process produces one Schur complement after another. Recall that a Schur complement of a Schur complement is a Schur complement. Then we looked for a zero 2-minor by 
going over all possible 2-minors in successive Schur complements. Note, Equation~\ref{eqn1} gives us the zero-minor in the original matrix corresponding to the zero 2-minor in the Schur complement.
If we found one we stopped. The data thus obtained is presented below as Table~\ref{tab:sometab}. The reader will notice a remarkable speed 
up. The number of kernels that we had to compute went down significantly as we started using Schur complements. Here the average is over ten tries.

\begin{table}[!h]
\centering
\caption{ Schur Complement : Average number of iteration to solve ECDLP.} \label{tab:sometab}
\vspace{2ex}
\begin{tabular}{l|c|c|c}
\hline
$\mathbb{F}_q$ & Average Iterations & $\log_2{p}$ & No.~of rows of $\mathcal{A}$  \\
\hline
$2^{40}$ & 2.8 & 39 & 360   \\
$2^{41}$ & 3.1 & 40 & 369   \\
$2^{42}$ & 6.6 & 41 & 378   \\
$2^{43}$ & 11.3 & 42 & 387 \\
$2^{44}$ & 17.1 & 42 & 396 \\
$2^{45}$ & 53.8 & 44 & 405 \\
$2^{46}$ & 40.2 & 45 & 414 \\
$2^{47}$ & 126  & 46 & 423 \\
$2^{48}$ & 186.5 & 47 & 432  \\
$2^{49}$ & 367.8 & 48 & 441 \\
$2^{50}$ & 887.5 & 49 & 450  \\
\hline
\end{tabular}
\end{table}
This data baffled us for a while. There are two mysteries. The first one is: the number of minors is exponential in the size of the matrix. If we assume that a zero-minor is uniformly distributed in the set of all minors, there is no way to justify the low number of tries -- both for all 2-minors and Schur complement. The other is: why do we see such a significant speed-up when we move from all 2-minors to Schur complements?

To make a fair comparison, we must admit that the Schur complement algorithm is doing more work than the all 2-minors algorithm. In particular, the Schur complement is doing multiple rounds of all 2-minors algorithm, once after each column reduction. So, let us count the number of all 2-minors algorithms done by a Schur complement algorithm. To make it precise, we count the area of the Schur complement $H^\prime$ on which the all 2-minor algorithm works. The sum of the areas is
\[(k/2)^2+(k/2+1)^2+\ldots+(k-1)^2=k(7k^2-9k+2)/24.\]
Then the ratio of the sum of the area for Schur complement and all 2-minors is 
$(7k^2-9k+2)/24k$. For $q=2^{46}$ where $k=414$, the ratio is about 120. To us this is encouraging news, it matches the experimental result. Moreover the pattern seems to be tied up with column reduction which is used in computing the Schur complement. If there were no zero-minors to be found, more work would have not produced better result.

A quick look at the above data suggests that it could be modeled by a recurrence relation that approximately doubles the previous number of average iterations at every step. This kind of recurrence can be solved and the solution is exponential in the number of steps. It then became clear that this attack cannot be extended to much higher values of $q$ and we decided to stop the experiment. 
\section{Conjecture -- initial minors}
The idea of initial minor is not new in the theory of matrices. It is
normally studied in the context of total positive matrices defined over
$\mathbb{R}$. We cannot discuss it further in this paper.  However, a
interested reader can look at Ando~\cite{ando}, Gasca and
Pen\={a}~\cite{gasca} or Pinkus~\cite{pinkus}. The reason of our interest in
that theory is that we can determine if a given matrix is total positive by
just looking at a few (quadratic in the size of the matrix) minors.

We think that something similar is true for $\mathcal{A}$ that we defined
earlier. One way to look at the issues that we raised before is that there
is a small set of minors of $\mathcal{A}$, such that, if all minors in that
set is non-zero then all minors of $\mathcal{A}$ is non-zero. In other
words, if there is a zero-minor of $\mathcal{A}$ then one of the minors in
the set of initial minors will be zero.  By doing Schur complements, we
probably found some neighbourhood of that set of initial minors. But, we did
not find the whole set. In view of our experimental evidence, we would like
to promote the \emph{conjecture of initial minors} (conjecture A), to
solve the elliptic curve discrete logarithm problem. 

\section{Implementing the Las Vegas algorithm}
We now describe a implementation of our algorithm using the Schur complement which is  available in the bitbucket repository~\cite{code}. The source code provides both a serial and a parallel implementation. Solving ECDLP over prime and binary fields is possible using our
implementations. The code is written using C++ and uses NTL and openMPI libraries. The following classes
\texttt{EC\_GF2E} and \texttt{EC\_ZZp} were implemented for elliptic curve operations, see Appendix A.
The program is developed in such a way that it reads input from a
file, which is generated by a SageMath~\cite{sage} script.
\subsection{The Las Vegas algorithm with Schur complements}
\begin{figure*}[tpb]
\begin{tikzpicture}[node distance=1.3cm]
\node (start) [startstop] {Start};
\node (pro1) [process, below of=start] {\;\; Gnerate $N_P\cdot 2k$ random numbers \;\;};
\node (pro2) [process, below of=pro1] { \;\; For $i = 0, 1,\ldots,N_p - 1$ construct  $\mathcal{M}$ and the left-kernels $\mathcal{K}$ and store \;\; };
\node (pro2b) [process, below of=pro2] { \;\; For $i = 0, 1,\ldots, N_P - 1$ pick a kernel $\mathcal{K}$ and extract $\mathcal{A}$   \;\; };
\node (pro3) [process, below of=pro2b] { \;\; For $j = 1, 2,\ldots,k$ reduce $j^\textsuperscript{th}$ column of $\mathcal{A}$, extract $H^\prime$ and distribute \;\; };

\node (pro4a) [SubProcess, below of=pro3] { \;\; $\ldots$ \;\; };
\node (pro4c) [SubProcess, below of=pro3, left of=pro4a, yshift=1.30cm, xshift=-0.5cm ] { \;\; $P_0$ \;\; };
\node (pro4d) [SubProcess, below of=pro3, right of=pro4a, yshift=1.30cm, xshift=0.85cm] { \;\; $P_{N_{P-1}}$ \;\; };
\node (pro5) [process, below of=pro4a, xshift=0cm] { \;\; If a zero 2-minor is detected by any processor STOP \;\; };

\node (stop) [startstop, below of=pro5] {Repeat};

\draw [arrow] (start) -- (pro1);
\draw [arrow] (pro1) -- (pro2);
\draw [arrow] (pro2) -- (pro2b);
\draw [arrow] (pro2b) -- (pro3);
\draw [arrow] (pro3) -- (pro4a);
\draw [arrow] (pro3) -- (pro4c);
\draw [arrow] (pro3) -- (pro4d);

\draw [arrow] (pro4a) -- (pro5);
\draw [arrow] (pro4c) -- (pro5);
\draw [arrow] (pro4d) -- (pro5);

\draw [arrow] (pro5) -- (stop);

\end{tikzpicture}
\caption{Flow chart of a parallel implementation of our algorithm.}
\end{figure*}
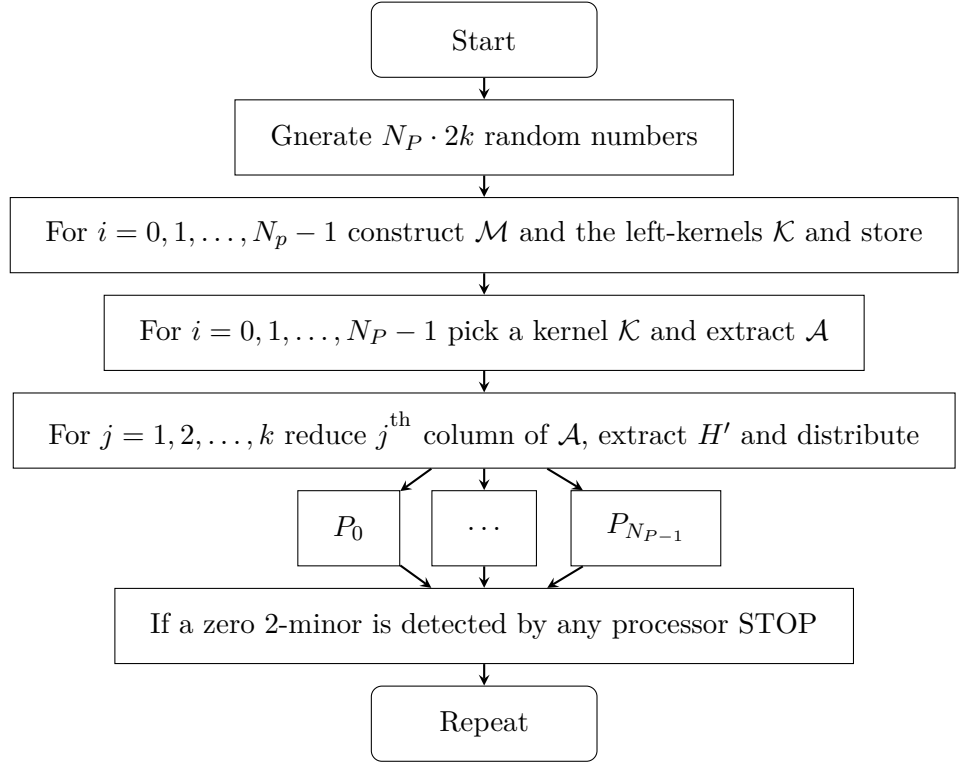	
Note that the index-origin for the list of kernels is 0, whereas the index-origin for rows and columns of a matrix is considered 1.

We assume a total of $N_P$ available processors.
We use the word processor instead of core.
The processor with \texttt{MPI\_rank} zero is assumed to be the \texttt{MASTER} processor.
All other processors are \texttt{SLAVE} processors.
All processors have access to the shared memory.

First the Las Vegas algorithm generates a set of random numbers.
These random numbers are used to construct matrices $\mathcal{M}$.
Then $\mathcal{K}$ = left-kernel($\mathcal{M}$) is computed and $\mathcal{A}$ is extracted.
A vector with $k$ zeros in $\mathcal{K}$ solves ECDLP.
\begin{description}
\item[1] In the first step the \texttt{MASTER} processor generates $2k$ random numbers for each $N_P$ processors and stores them in $N_P$ files. All processors have access to these files.
Thus there are $N_P$ files in a directory named randomNumbers.
At this point all \texttt{SLAVE} processors are idle.

\item[2] In the next step each $N_P$ processor reads their file containing random numbers and construct the matrix $\mathcal{M}$. Then each processor computes the left kernel $\mathcal{K}$ of $\mathcal{M}$.
To compute the left-kernel an inbuilt function from NTL is used.
Note that at this step $N_P$ number of kernels are generated.
Each processor stores the computed kernel in a separate files.
These are stored in a directory named kernel. Later these files are systematically processed.

\item[3] In the third step the \texttt{MASTER} processor extracts $\mathcal{A}$ from a kernel from the files. Columns of the kernel is then reduced iteratively using an inbuilt NTL function. After each
  iteration $H^\prime$ is extracted.
For the $i^\textsuperscript{th}$ iteration $H^\prime$ is a $(k-i) \times (k-i)$ matrix as $\mathcal{A}$ is a $k \times k$ matrix.
The matrix $H^\prime$ along with the combination of rows to be processed is distributed among the \texttt{SLAVE} processors.
Now each processor looks for a zero 2-minor.
\end{description}
In order to detect a 2-minor we select two rows and two columns from the given matrix $H^\prime$.
The intersection of these two rows and columns give a $2 \times 2$ matrix.
If the determinant of this $2 \times 2$ matrix is zero we have found our zero 2-minor.
To detect a zero two minor we take a different approach.
We select two rows from the given matrix.
Then the ratios of corresponding elements of these two rows are computed and stored in a vector.
If two entries in this vector are equal we have detected a two minor. Note that we are dealing with characteristic 2 in this case.

To implement this method in a parallel environment, combinations of all rows have to be processed.
We divide the total number of combinations of rows by the number of available processors.
Each processor then processes a subset of the total combinations.
If a two minor is detected we stop and row and column indices of the two minor is returned. Using Equation~\ref{eqn1} we can now say that there is a zero-minor in $\mathcal{A}$ along with the rows and columns.

An attentive reader might have noticed, in step two of our algorithm all processors are computing a kernel.
This approach is taken as computing the kernel at the \texttt{MASTER} processor leaves the 
\texttt{SLAVE} processors idle.
Also kernel computation is time consuming. These keeps the \texttt{SLAVE} processors waiting, increasing the overall execution time.
Computing kernels on all processors and saving them in files for later use reduces processor wait time.
The number of kernels generated using this method depends on the number of available processors.
\section{Conclusion}
This work presents a new way to solve the elliptic curve discrete logarithm problem. The problem is reduced to a linear algebra problem. This reduction to the linear algebra problem is fast. Then the linear algebra problem is reduced to finding a zero-minor in a non-singular square matrix over a finite field. This is the novelty of our approach. We found a way to look for a zero-minor and conjectured initial minor. If there is a small set of initial minors, polynomial in size of the group of rational points of the elliptic curve, we will have a polynomial time algorithm to solve ECDLP.
\bibliographystyle{plain}
\bibliography{paper}

\begin{thebibliography}{1}

\bibitem{code}
Ansari Abdullah.
\newblock {\em LasVegas-ECDLP}, 2019.
\newblock https://bitbucket.org/abdullah0096/lasvegas-ecdlp.git.

\bibitem{ando}
T.~Ando.
\newblock Totally positive matrices.
\newblock {\em Linear algebra and its applications}, 90:165--219, 1987.

\bibitem{bru}
Richard~A. Brualdi and Hans Schneider.
\newblock Determinantal identities: Gauss, {Schur}, {Cauchy}, {Sylvester},
  {Kronecker}, {Jacobi}, {Binet}, {Laplace}, {Muir}, and {Cayley}.
\newblock {\em Linear algebra and its applications}, 52/53:769--791, 1983.

\bibitem{icm}
Sergey Fomin.
\newblock Total positvity and cluster algebras.
\newblock In {\em International Congress of Mathematicians Hyderabad, India
  2010 ({ICM 2010})}, volume~2, pages 125 -- 145, 2011.

\bibitem{gasca}
M.~Gasca and J.\thinspace{}M. Pen\={a}.
\newblock Total positivity and {N}eville elimination.
\newblock {\em Linear algebra and its applications}, 165:25--44, 1992.

\bibitem{one}
Ayan Mahalanobis, Vivek~M. Mallick, and Ansari Abdullah.
\newblock A {Las Vegas} algorithm to solve the elliptic curve discrete
  logarithm problem.
\newblock In {\em Progress in Cryptology -- INDOCRYPT2018}, volume 11356 of
  {\em LNCS}, pages 215--227, 2018.

\bibitem{pinkus}
Alan Pinkus.
\newblock Zero minors of total positive matrices.
\newblock {\em Electronic Journal of Linear Algebra}, 17:532--542, 2008.

\bibitem{Shoup1996NtlAL}
Victor Shoup.
\newblock {NTL}: a library for doing number theory.
\newblock 1996.

\bibitem{sage}
W.\thinspace{}A. Stein et~al.
\newblock {\em {S}age {M}athematics {S}oftware ({V}ersion 9.0)}.
\newblock The Sage Development Team, 2020.
\newblock {\tt http://www.sagemath.org}.

\end{thebibliography}
\appendix
\section{ }
\begin{lstlisting}[basicstyle=\ttfamily\scriptsize]
class EC_GF2E{
public:
    ulong p;
    GF2X irrd;
    GF2E a1, a2, a3, a4, a6;
    GF2E discriminant;

    EC_GF2E(ulong);
    EC_GF2E(ulong, GF2X, GF2X);
    EC_GF2E(ulong, GF2X, GF2X, GF2X);

    void generateRandomCurve();
    void printCurve();
    EC_GF2E_Point generateRandomPoint();
    bool isPointValid(const EC_GF2E_Point &);
    GF2E getDiscriminant(const GF2E &, const GF2E &);

    void pointAddition_Doubling(const EC_GF2E_Point&, const EC_GF2E_Point&, EC_GF2E_Point &);
    EC_GF2E_Point pointDoubling(const EC_GF2E_Point &P);    
    void pointNegation(const EC_GF2E_Point&, EC_GF2E_Point&);
    void scalarMultiplicationDA(const EC_GF2E_Point&, ZZ, EC_GF2E_Point &);
    int generateMatrix(mat_GF2E &, EC_GF2E_Point, EC_GF2E_Point, ulong, ulong, ZZ *, ulong a[][3]);

    ZZ lasVegasECDLP(EC_GF2E_Point &, EC_GF2E_Point &, ZZ ordP, int);

    friend EC_GF2E_Point operator+(const EC_GF2E_Point &P, const EC_GF2E_Point & Q);
};
\end{lstlisting}

\begin{lstlisting}[basicstyle=\ttfamily\scriptsize]
class EC_ZZp{
public:

    ZZ p;
    ZZ_p a4, a6;
    ZZ_p discriminant;

    EC_ZZp(ZZ);
    EC_ZZp(ZZ, ZZ, ZZ);
    EC_ZZp(ZZ, ZZ, ZZ, ZZ);

    void generateRandomCurve();
    EC_ZZp_Point generateRandomPoint();

    void printCurve();
    bool isPointValid(const EC_ZZp_Point &P);
    void pointAddition_Doubling(const EC_ZZp_Point&, const EC_ZZp_Point&, EC_ZZp_Point &);

    void scalarMultiplication_Basic(const EC_ZZp_Point&, ZZ, EC_ZZp_Point &);
    void scalarMultiplicationDA(const EC_ZZp_Point&, ZZ, EC_ZZp_Point &);

    void pointNegation(const EC_ZZp_Point &, EC_ZZp_Point &);
    ZZ order(const EC_ZZp_Point&);

    int generateMatrix(mat_ZZ_p &, EC_ZZp_Point P, EC_ZZp_Point Q, \
        ulong k_randomNums, ulong t_randomNums, ZZ *PQ_randomNumbers, ulong weightedVector_arr[][3]);
           
    ZZ lasVegasECDLP(const EC_ZZp_Point &P, const EC_ZZp_Point &Q, ZZ);    
};
\end{lstlisting}

\end{document}